\documentclass[12pt]{iopart}
\usepackage{iopams}

\newenvironment{pmatrixTWO}
  {\left(\begin{array}{cc}} 
  {\end{array}\right)}

\usepackage[pdftex]{graphicx}
\usepackage{hyperref}
\usepackage{overpic}
\usepackage{diagbox}
\usepackage{wasysym}
\usepackage{float}
\usepackage{comment}
\usepackage[section]{placeins}
\usepackage{subcaption}
\usepackage{bm}

\newtheorem{defin}{Definition}

\newtheorem{theorem}{Theorem}

\newtheorem{proposition}{Proposition}
\begin{document}

\title[Unique determination and generation of rank two single-qubit quantum channels]{Unique determination and generation of rank two single-qubit quantum channels\footnote{This manuscript is the author submitted version. The Version of Record has a significantly polished numerical part and is published as:
N.E.~Kuklina, V.Iu.~Vasilenko, B.O. Volkov, and A.N. Pechen, Unique determination and generation of rank-two single-qubit quantum channels, Journal of Physics A: Mathematical and Theoretical, 59, 355309 (2026).
https://doi.org/10.1088/1751-8121/ae9b4f}}

\author{Nina E. Kuklina$^{1}$, Vasilii Iu. Vasilenko$^{1}$, Boris O. Volkov$^{1}$,  and Alexander N. Pechen$^{1,2}$}
\address{$^1$ Department of Mathematical Methods for Quantum Technologies, Steklov Mathematical Institute of Russian Academy of Sciences, Gubkina str. 8, 119333, Moscow, Russia}
\date{}

\begin{abstract} 
An important problem in quantum technologies is the generation of a target evolution of an open quantum system. A core component of this task is the ability to determine whether the actual evolution of the system coincides with the target evolution. For the generation of unitary quantum channels in open quantum systems, it was shown by Goerz, Reich and Koch [New J. Phys. {\bf 16} 055012 (2014)] that for determining whether the actual evolution coincides with a desired unitary it is sufficient to compare their action on three special density matrices. In this work, we consider controlled generation and unique determination of non-unitary quantum channels in open quantum systems with particular emphasis on rank two single-qubit quantum channels. We prove that to uniquely determine such quantum channels it is sufficient to consider their action on only three suitably chosen density matrices. Based on this theoretical result, we numerically investigate generation of various non-unitary target single-qubit quantum channels in open quantum systems using coherent and incoherent controls.
\end{abstract}


\section{Introduction}  

The ability to accurately control quantum systems is necessary for diverse directions in quantum technologies~\cite{KochEPJQuantumTechnol2022,TannorBook2007,ShapiroBrumerBook2012,AlessandroBook2021,Ma_Qi_Petersen_Wu_Rabitz_Dong_2025}. In many applications, controlled quantum systems are open, i.e. interacting with a surrounding environment. This circumstance motivates the importance of developing methods for controlling open quantum systems. The most general evolution of an open quantum system is described by a quantum channel, i.e., completely positive trace-preserving (CPTP) linear map, also called a Kraus map~\cite{Kraus1983, Breuer_Book_1993, Nielsen_Chuang_2010, Wilde_Book_2017, Holevo_Book_2019} (non-completely positive maps are also considered in some works~\cite{Pechukas_1994,Shaji_Sudarshan_2005}). 

Open quantum system dynamics with non-unitary evolution has been explored across a wide range of quantum control and information-processing settings. Its relevance was established for quantum computing with mixed states and non-unitary quantum gates~\cite{Aharonov1998,Tarasov_2002}. Measurement-assisted quantum control based on the Strocchi map was studied~\cite{VilelaMendes2003}. The role of the environment as a resource was explored through incoherent control schemes~\cite{Pechen2006a}, and later was applied for preparation of many-body entangled states via engineered dissipation~\cite{Diehl2008} and to dissipation-driven approaches for improving quantum computation~\cite{Verstraete2009}. Non-unitary all-to-one, or replacement, quantum channels were proposed for robust quantum control~\cite{Wu_Pechen_Brif_Rabitz_2007}. Experimental progress demonstrated the induction of multiparticle entanglement dynamics through controlled decoherence~\cite{Barreiro2010}, theoretical advances on using pointer states stabilized by engineered environments~\cite{Khodjasteh2011}, optimal control strategies exploiting cooperative effects of driving and dissipation~\cite{Schmidt2011}, deterministic entanglement generation in Rydberg ensembles~\cite{Rao2014} and incoherent control of molecular processes such as retinal isomerization~\cite{Lucas2014}, as well as precise manipulation of qubits via measurement backaction and feedback~\cite{Blok2014}. More recently, open-system control techniques have been applied to tailoring kinetic processes in complex materials~\cite{Laforge2018} and to implementing experimentally accessible Kraus-map-based replacement (all-to-one) transformations such as polarization control~\cite{Zhang2022}. 

Characterization of properties of general non-unitary quantum channels plays an important role in quantum information theory and its applications, including quantum tomography, information transfer, and quantum control.  An important problem is to find a minimal number of quantum states (density matrices) which allow to uniquely determine whether an actual evolution coincides with some quantum channel from a given set of quantum channels (e.g., with a unitary quantum channel, rank two quantum channel, etc.) or not. That is, for a quantum channel $\Phi^*$ from the given set of quantum channels and for an arbitrary quantum channel $\Phi$, how many different density matrices $\rho_1,\dots,\rho_k$ can be used to claim that $\Phi=\Phi^*$ if $\Phi(\rho_1)=\Phi^*(\rho_1),\dots, \Phi(\rho_k)=\Phi^*(\rho_k)$? If $\Phi^*$ is a unitary quantum channel, then as was shown in~\cite{Goerz_NJP_2014_2021,Goerz_2021}, for an $N$-level quantum system it is sufficient to set $k=3$ and choose some special density matrices $\rho_1,\rho_2,\rho_3$. Based on these results, generation of single-qubit and two-qubit unitary channels in open quantum systems using coherent and incoherent drives was investigated~\cite{PetruhanovPhotonics2023,PechenPetruhanovMorzhin}. 

A natural question is to consider other that unitary classes of quantum channels. In general it turns out to be a nontrivial problem. In this work, we show that a general quantum channel $\Phi: \mathbb{C}^{n \times n} \to \mathbb{C}^{m \times m}$ (with $m>1$) cannot be uniquely determined by fewer than $n^2$ input states, provided its Choi matrix has full rank $nm$. Further, we prove that for every rank-two single qubit quantum channel there exists a set of three states which uniquely determines it, so that in this case $k=3$. Rank two quantum channels can be considered as simplest beyond unitary which are rank one. Finding the minimal number of states necessary to uniquely determine a quantum channel is important for construction of a proper objective functional which can be used to generate target non-unitary processes in numerical setups. It is even more important in experimental setups to minimize cost of overall quantum channel tomography, which is cost of tomography of one quantum state multiplied by the number of quantum states.

Based on this theoretical result, we investigate generation of various non-unitary target single-qubit quantum channels, such as replacement and phase damping channels, using coherent and incoherent controls of an open quantum system and incoherent GRAPE (inGRAPE) method~\cite{PetruhanovPechenJPA2023}. As one may expect, not all target quantum channels can be generated with high precision, since for the problem of generation of non-unitary quantum channels even such basic problem as controllability criteria for a particular physical model is generally not studied. Although controllability in the sets of density matrices for open quantum systems was studied in some cases~\cite{Altafini1,Altafini2,Wu_Pechen_Brif_Rabitz_2007} and the explicit structure of the reachable sets for an open qubit was established~\cite{Lokutsievskiy_Pechen_2021}, controllability on the set of quantum channels still requires its investigation.

The paper is organized as follows. In Sec.~\ref{Unique_determination_Sec}, we introduce a notion of a set of states that uniquely determines a quantum channel and prove that a full-rank channel cannot be uniquely determined by a set of states that is smaller than the square of the input space dimension (Theorem~1). Sec.~\ref{Three_states_Sec} contains proof of the main theorem that any single-qubit rank two quantum channel can be uniquely determined by a set of three quantum states (Theorem~2). In Sec.~\ref{Problem_of_generating_Sec}, we formulate the problem of generating a (generally non-unitary) single-qubit quantum channel using coherent and incoherent controls. In Sec.~\ref{sec_inGRAPE}, the gradient-based inGRAPE method for numerical generation of general quantum channels is described. Sec.~\ref{Sec_numerical} provides the obtained numerical results for generation of replacement channels and a phase damping channel. In~\ref{Parametrization}, we describe the parametrization of the dynamical equations and objective functionals. \ref{AppendixGradient} contains analytical formulas for the gradients of the objective functionals used for the numerical generation of quantum channels. Conclusions and Discussion Sec.~\ref{Conclusions} summarizes this work.

\section{Unique determination of quantum channels by their action on sets of states}
\label{Unique_determination_Sec}

A linear map $\Phi: \mathbb{C}^{n \times n} \rightarrow \mathbb{C}^{m \times m}$ is called a \textbf{quantum channel} if it is trace-preserving and completely positive. That is, for any $A\in \mathbb{C}^{n \times n}$ one has  $\Tr(\Phi(A))= \Tr(A)$, and for every $k\in \mathbb{N}$ the map $\Phi \otimes \mathbb{I}_{k}\colon \mathbb{C}^{n \times n}\otimes \mathbb{C}^{k \times k}\to
\mathbb{C}^{m \times m}\otimes \mathbb{C}^{k \times k}$ transforms positive-semidefinite operators to positive-semidefinite operators. Here $\mathbb{I}_{k}: \mathbb{C}^{k \times k} \rightarrow \mathbb{C}^{k \times k}$ denotes the identity operator. We denote the set of quantum channels from $\mathbb{C}^{n \times n}$ to $\mathbb{C}^{m \times m}$ by $\mathcal{C}han(n,m)$. 
Equivalently, by the Kraus representation theorem, a linear map $\Phi:\mathbb{C}^{n \times n} \rightarrow \mathbb{C}^{m \times m}$ is a quantum channel \cite{Kraus1983} if and only if there exist operators $\{K_j\}_{j=1}^{l} \subset \mathbb{C}^{m\times n}$ such that for every density matrix $\rho$
\begin{equation}\label{kraus}
    \Phi(\rho) = \sum\limits_{j=1}^{l}K_j\rho K_j^\dagger, \quad \sum\limits_{j=1}^{l}K_j^\dagger K_j = \mathbb{I}_n.
\end{equation}
This representation of a given quantum channel is non-unique, since any set of Kraus operators $\tilde K_i=\sum\limits_j U_{ij}K_j$ with $U_{ij}$ being matrix elements of some unitary matrix $U$ determines the same quantum channel. The minimal possible number of operators $l$ in the decomposition (\ref{kraus}) is called the Kraus rank of the channel $\Phi$. 
 
Let $\Phi:\mathbb{C}^{n \times n} \rightarrow \mathbb{C}^{m \times m} $  be a linear map and let  $E_{ij}$ be the $n \times n$ matrix with 1 in the $(i,j)$-th entry and 0 elsewhere. Then $\mathcal{C}_{\Phi} = \sum\limits_{i, j=1}^{n} E_{ij} \otimes \Phi(E_{ij})$ is called the Choi matrix of the map $\Phi$.  The map $\Phi$ is completely positive if and only if $\mathcal{C}_\Phi$ is positive semidefinite \cite{choi75}. The rank of the Choi matrix of a quantum channel is equal to its Kraus rank. Another way to treat quantum channels is to consider them as linear operators between spaces of Hermitian matrices of the appropriate dimensions. Fix orthonormal bases for the spaces of Hermitian $n \times n$ and $m \times m$ matrices, consisting of the identity operator together with a basis of zero-trace Hermitian matrices, for instance, the generalized Gell-Mann matrices (see, e.g.,~\cite{bloch}). With respect to these bases, any quantum channel $\Phi\in \mathcal{C}han(n,m)$ is represented by a real matrix
\begin{equation}\label{matrix}
    \mathbb{T}_\Phi  = \left(\begin{array}{cc}
1 & \textbf{0}_{1\times3} \\
\mathbf{s} & B
\end{array}\right).
\end{equation}
Here $B \in \mathbb{R}^{(m^2-1)\times(n^2-1)}$, % $\textbf{0}$ is the zero row vector $1 \times (n^2-1)$,
$\textbf{s}$ is the column vector $(m^2-1) \times 1$. 
We now introduce the main definition.

\begin{defin}
    Let $\Phi\in \mathcal{C}han(n,m)$ be a quantum channel and let $ \{\rho_j\}_{j=1}^d$ be a set of $n\times n$ density matrices. If for every quantum channel $\Psi \in \mathcal{C}han(n,m)$ satisfying $\Psi(\rho_j)=\Phi(\rho_j)$ for all $j\in\{1,\ldots,d\}$, it follows that $\Psi=\Phi$, then we  say that the quantum channel $\Phi$ is {\bf uniquely determined} by the set of states $\{\rho_j\}_{j=1}^d$.     
\end{defin}
Any quantum channel $\Phi \in \mathcal{C}han(n,m)$ can be uniquely determined by a set of $n^2$ linearly independent quantum states. The ability to uniquely determine a channel with fewer states allows us to construct objective functionals for quantum control problems on a reduced set of states.

We first recall some previously known results.
In~\cite{Minnumber, Goerz_NJP_2014_2021} it was shown that a unitary quantum channel acting on an arbitrary finite-dimensional space can be uniquely determined by only three quantum states. The replacement channel provides an example of a channel that can be uniquely determined by even fewer states. Consider the replacement channel that maps every input state to a fixed pure state $\rho$. That is, for any density matrix $\omega$, the following holds:
\begin{equation*}
    \Phi_\rho(\omega) := \rho.
\end{equation*}
Indeed, the only quantum channel that maps a state whose spectrum does not contain zero to a pure state $\rho$ is the replacement channel $\Phi_\rho$ \cite{Kohout2010}. Therefore, if $\rho$ is pure, then $\Phi_\rho$ is uniquely determined by a single state. 

Denote by $\mathcal{L}(\mathbb{C}^{n \times n}, \mathbb{C}^{m \times m} )$ the space of linear maps from $n \times n$ to $m \times m$ complex matrices and by $\mathcal{L}_{\mathrm{TP}}(\mathbb{C}^{n \times n}, \mathbb{C}^{m \times m})$ the affine subspace of $\mathcal{L}(\mathbb{C}^{n \times n}, \mathbb{C}^{m \times m} )$ consisting of trace-preserving maps. We will use three equivalent representations of quantum channels: as abstract maps (by $\Psi \in   \mathcal{L}(\mathbb{C}^{n \times n}, \mathbb{C}^{m \times m} )$), as matrices of linear maps between spaces of Hermitian matrices (by $\mathbb{T}_\Psi \in \mathbb{R}^{m^2 \times n^2}$), and as Choi matrices (by $\mathcal{C}_\Psi \in \mathbb{C}^{mn\times mn}$). Let $\Phi \in \mathcal{C}han(n,m)$ be a quantum channel and let $\mathcal{S}$ be a set of quantum states. Denote
 \begin{equation*}
    \mathcal{P}(\Phi,\mathcal{S}):=\{\Psi \in \mathcal{L}_{\mathrm{TP}}(\mathbb{C}^{n \times n}, \mathbb{C}^{m \times m}) \mid \Psi(\rho)=\Phi(\rho) \mbox{ for every } \rho \in \mathcal{S}\}.
\end{equation*}
In the $\mathbb{T}$-representation, this becomes
\begin{equation*}
    \mathcal{P}_{\mathbb{T}}(\Phi,\mathcal{S}):=\{\mathbb{T}_\Psi |\Psi \in \mathcal{P}(\Phi,\mathcal{S})\}.
\end{equation*}
Under the Choi isomorphism, this corresponds to
\begin{equation*}
    \mathcal{P}_{\mathcal{C}}(\Phi,\mathcal{S}):=\{ \mathcal{C}_\Psi | \Psi \in \mathcal{P}(\Phi,\mathcal{S}) \}.
\end{equation*}
\begin{theorem}
    A quantum channel $\Phi \in \mathcal{C}han(n,m)$ with $m >1$ and with $\rank \Phi = nm$, cannot be uniquely determined by less than $n^2$ states.
\end{theorem}
\noindent{\bf Proof.}    By construction, $\mathcal{P}_\mathbb{T}(\Phi, \mathcal{S})$ is an affine subspace of $\mathbb{R}^{m^2 \times n^2}$. The trace-preserving condition imposes $n^2$ linearly independent constraints by fixing the first row of matrices in $\mathcal{P}_\mathbb{T}(\Phi, \mathcal{S})$ to be $(1, \mathbf{0})$.  For each $\rho \in \mathcal{S}$, the condition $\Psi(\rho) = \Phi(\rho)$ imposes $m^2-1$ linear constraints.
Hence, the dimension of the affine subspace satisfies
    \begin{equation*}
       \dim \mathcal{P}_\mathbb{T}(\Phi, \mathcal{S}) \geq (nm)^2-n^2-|\mathcal{S}|(m^2-1),
    \end{equation*}
where $|\mathcal{S}|$ is the number of states in $\mathcal{S}$. Since $|\mathcal{S}|<n^2$, we obtain
    \begin{equation*}
        \dim \mathcal{P}_\mathbb{T}(\Phi, \mathcal{S}) > (nm)^2-n^2-(nm)^2+n^2=0
    \end{equation*}
Thus, $\dim_\mathbb{R} \mathcal{P}_\mathcal{C}(\Phi, \mathcal{S}) \geq 1$ and $\mathcal{P}_\mathcal{C}(\Phi, \mathcal{S})$ contains a line through $\mathcal{C}_\Phi$. The set of states $\mathcal{S}$ uniquely determines $\Phi$ if and only if 
\begin{equation*}
    \mathcal{P}_{\mathcal{C}}(\Phi,\mathcal{S}) \cap S_+(\mathbb{C}^{nm}) = \{ \mathcal{C}_\Phi\},
\end{equation*}
where $S_+(\mathbb{C}^{nm})$ denotes the cone of positive semidefinite $nm \times nm$ complex matrices. Since $\rank \mathcal{C}_\Phi= nm$, the matrix $\mathcal{C}_\Phi$ is positive definite, i.e., it lies in the interior of $S_+(\mathbb{C}^{nm})$. Hence, there exists an open ball $B_\varepsilon(\mathcal{C}_\Phi)$ centered at $\mathcal{C}_\Phi$ with radius $\varepsilon > 0$ that is entirely contained in $S_+(\mathbb{C}^{nm})$. The ball $B_\varepsilon(\mathcal{C}_\Phi)$ intersects with $\mathcal{P}_{\mathcal{C}}(\Phi,\mathcal{S})$  in infinitely many points (e.g., its intersection with $\mathcal{P}_{\mathcal{C}}(\Phi,\mathcal{S})$ contains a line segment). All these points lie in $S_+(\mathbb{C}^{nm})$ because the ball is contained in the cone. Hence, $\mathcal{P}_{\mathcal{C}}(\Phi,\mathcal{S}) \cap S_+(\mathbb{C}^{nm}) \neq \{ \mathcal{C}_\Phi \}$, which means that $\mathcal{S}$ does not uniquely determine $\Phi$. This completes the proof.

\section{Three states sufficient for a unique determination of rank-two single-qubit quantum channels}
\label{Three_states_Sec}

In this section, we consider single-qubit quantum channels. Fix in the space of  $2 \times 2$ Hermitian matrices  a basis $\{\mathbb{I}, \sigma_1, \sigma_2, \sigma_3\}$, where $\mathbb{I}$ is the identity matrix, and $\sigma_1, \sigma_2, \sigma_3$ are the Pauli matrices:
$$
\sigma_1=\sigma_x = \left(\begin{array}{cc}
0 & 1 \\
1 & 0
\end{array}\right),\quad
\sigma_2=\sigma_y = 
\left(\begin{array}{cc}
0 & -i \\
i & 0
\end{array}\right),\quad
\sigma_3=\sigma_z = \left(\begin{array}{cc}
0 & 1 \\
1 & 0
\end{array}\right).
$$
Note that for $2 \times 2$ matrices, the set of generalized Gell-Mann matrices  coincides with $\{\sigma_1, \sigma_2, \sigma_3\}$. Here, we refer to the  basis $\{\mathbb{I}, \sigma_1, \sigma_2, \sigma_3\}$, as the standard basis. Any density matrix $\rho$ can be represented in the standard basis as
\begin{equation}
\label{Bloch_parametrization}
    \rho = \frac{1}{2}\bigl( \mathbb{I} + \sum\limits_{i=1}^3 r_i \sigma_i \bigr),
\end{equation} 
where $\mathbf{r}=(r_1,r_2,r_3)^T \in \mathbb{R}^{3}$ is the Bloch vector of the state $\rho$. All such Bloch vectors form the unit ball, i.e., $|\mathbf{r}| \leq 1$.  The equality $|\mathbf{r}|=1$ holds precisely for pure states.
In the standard basis, any quantum channel $\Phi\in \mathcal{C}han(2,2)$ is represented by a real matrix~(\ref{matrix}), where  $B$ is a matrix $3 \times 3$ , $\mathbf{s}$ is a column vector of length $3$, and $\mathbf{0}$ is a row vector of length $3$.
 In~\cite{MinEn_King_Russki} it was shown that any single qubit quantum channel $\Phi$ can be decomposed into a map $\Phi_0$ preceded and followed by unitary conjugations, where $\Phi_0$ is represented in the standard basis by a matrix of the form
\begin{equation} \label{diag}
    \mathbb{T}_{\Phi_0} = \left(\begin{array}{cccc}
1 & 0 & 0 & 0 \\
t_1 & \lambda_1& 0 & 0 \\
t_2 & 0& \lambda_2 & 0 \\
t_3 & 0& 0 & \lambda_3 
\end{array}\right).
\end{equation}
That is,
\begin{equation}\label{rotation}
    \Phi(\rho) = U \Phi_0 (V\rho V^{\dagger})U^{\dagger}.
\end{equation}
Here $U,V \in \mathrm{U}(2)$, and the conjugation $U(\cdot)U^\dagger$ corresponds to a rotation of the Bloch ball.

It was also shown in~\cite{RuskaiSzarekWerner} that qubit channels of rank one and rank two can be represented in the form (\ref{rotation}), where $\Phi_0$ is given by
\begin{equation}\label{matrixofrank2}
    \mathbb{T}_{\Phi_0}= \left(\begin{array}{cccc}
1 & 0 & 0 & 0 \\
0 & \cos{u}& 0 & 0 \\
0 & 0& \cos{v} & 0 \\
\sin{u} \sin{v} & 0& 0 & \cos{u} \cos{v} \\
\end{array}\right).
\end{equation}
Here $u \in [0, 2\pi)$ and $v \in [0, \pi)$.

We now consider qubit quantum channels of Kraus rank 2. Our aim is to show that any such channel can be uniquely determined by its action on three suitably chosen states. Due to the invariance property established in the following proposition, it is sufficient to study only channels of the form (\ref{matrixofrank2}), which throughout this section will be denoted as $\Phi_0$.

\begin{proposition}\label{unitinvar}
   Let $\Phi_1 \in \mathcal{C}han(2,2)$  be uniquely determined by the set of states $\{\rho_j\}_{j=1}^d$. Then the quantum channel $\Phi_2 =  U \Phi_1 V(\cdot) V^{\dagger}U^{\dagger}$, where $U,V \in \mathrm{U}(2)$, is uniquely determined by the set $\{V^\dagger\rho_i V\}_{j=1}^d$.
\end{proposition}
\noindent{\bf Proof.}
 Let the following equality hold for the quantum channel $\Psi$: 
    \begin{equation}\label{AeqPhi}
        \Psi(V^\dagger \rho_j V) = \Phi_2 (V^\dagger \rho_j V).
    \end{equation}
    Let $\Psi_1 = U^\dagger\Psi (V^\dagger (\cdot)V) U$. Taking into account (\ref{AeqPhi}), we obtain:
    \begin{equation*}
        \Psi_1\rho_j=U^\dagger\Psi (V^\dagger \rho_jV) U
        =U^\dagger\Phi_2 (V^\dagger \rho_jV) U
       =\Phi_1 \rho_j.
    \end{equation*}
Since $\Phi_1$ is uniquely determined by the set of states $\{ \rho_j\}_{j=1}^d$, it follows that $\Psi_1 = \Phi_1$. By the definitions of $\Psi_1$ and $\Phi_1$, the equality $\Psi=\Phi_2$ is satisfied. Therefore, $\Phi_2$ is uniquely determined by the set of density matrices $\{V^\dagger \rho_j V\}_{j=1}^d$. This completes the proof.

Consider the set of states $\mathcal{S}_3= \{\bm{\rho}_1, \bm{\rho}_2,\bm{ \rho}_3\}$, where
  \begin{equation}\label{three_states}
    \bm{\rho}_1 =\frac{1}{2} \left(\begin{array}{cc}
 1& 1\\
 1& 1
\end{array}\right), \quad
    \bm{\rho}_2 =\frac{1}{2}\left(\begin{array}{cc}
 1& -i\\
        i& 1
\end{array}\right), \quad
    \bm{\rho}_3 =\frac{1}{2}\left(\begin{array}{cc}
 1& 0\\
        0 & 1
\end{array}\right).
\end{equation}

\begin{proposition}\label{affinespacerho1rho2rho3}
   For an arbitrary $\mathcal{N} \in \mathcal{C}han(2,2)$, the set $ \mathcal{P}_{\mathbb{T}}(\mathcal{N},\mathcal{S}_3)$ is an affine space $\mathbb{T}_\mathcal{N} + W$, where 
   \begin{equation*}
      W = \left\{ \left(\begin{array}{cccc}
 0 & 0 & 0& 0\\
        0 & 0 & 0& c_1\\
        0 & 0 & 0 & c_2 \\
        0 & 0 & 0 & c_3
\end{array}\right) \middle| c_1, c_2,c_3 \in \mathbb{R} \right\}. 
   \end{equation*}
\end{proposition}
\noindent{\bf Proof.}
Let $X\in\mathbb{R}^{3\times 3}$ be the matrix whose $i$-th row is the Bloch vector of $\bm{\rho}_i \in \mathcal{S}$ for $i \in \{1, 2,3\}$.  Explicitly,
\begin{equation*}
    X = \left(\begin{array}{ccc}
        1 & 0& 0 \\
        0 & 1 & 0\\
        0 & 0& 0
\end{array}\right).
\end{equation*}
Consider $ \Psi \in \mathcal{P}(\mathcal{N},\mathcal{S})$. In the standard basis, the condition that $\Psi$ and $\mathcal{N}$ coincide on the states $\{\bm{\rho}_1, \bm{\rho}_2,\bm{ \rho}_3\}$ can be rewritten as
\begin{equation}\label{coincides}
    (\mathbb{T}_\Psi-\mathbb{T}_\mathcal{N})
    \left(\begin{array}{c}
 \mathbf{1}\\ X^T
\end{array}\right) = \textbf{0}_{4\times 3},
\end{equation}
where $\mathbf{1}$ denotes the row vector $(1,1,1)$.
From~(\ref{matrix}) it follows that the matrix $\mathbb{T}_\Psi- \mathbb{T}_\mathcal{N}$ can be represented in the form
\begin{equation*}
    \mathbb{T}_\Psi - \mathbb{T}_\mathcal{N} =  
    \left(\begin{array}{cccc}
 0 & 0&0&0\\
l_1 & D_{11}& D_{12}& D_{13}\\
l_2 & D_{21}& D_{22}& D_{23}\\
l_3 & D_{31}& D_{32}& D_{33}\\
\end{array}\right), 
\end{equation*}
where $l_i$ and $D_{ij}$ are real. For each row $i \in \{ 1,2,3\}$, equality~(\ref{coincides}) gives the condition
 \begin{equation*}
    X(D_{i1},D_{i2},D_{i3})^T = -l_i (1,1,1)^T.
\end{equation*}
Therefore, $(D_{i1},D_{i2},D_{i3})^T$ is a solution to the linear system $Xy=-l_i(1,1,1)^T$. If $l_i \neq 0$, then this system is inconsistent. Thus, $l_i = 0$ and  $(D_{i1},\dots,D_{i3})^T = c_i (0,0,1)$.
The difference $\mathbb{T}_\Psi - \mathbb{T}_\mathcal{N}$ is then forced into the form
\begin{equation*}
    \mathbb{T}_\Psi-\mathbb{T}_\mathcal{N}= 
    \left(\begin{array}{cccc}
  0 & 0 & 0& 0\\
        0 & 0 & 0& c_1\\
        0 & 0 & 0 & c_2 \\
        0 & 0 & 0 & c_3
\end{array}\right).
\end{equation*}
Conversely, any $\Psi$ constructed by adding such a matrix to $\mathbb{T}_\mathcal{N}$ coincides with $\mathcal{N}$ on the states $\bm{\rho}_1, \bm{\rho}_2, \bm{\rho}_3$. Therefore, the set $\mathcal{P}_\mathbb{T}(\mathcal{N}, \mathcal{S}_3)$ is precisely the affine space $\mathbb{T}_\mathcal{N}+W$, as defined above. This completes the proof.

By direct computations, we obtain the image of the linear space $W$ under the Choi isomorphism:
\begin{equation}\label{W}
\fl\qquad    \mathcal{C}_{W}=\left\{  \left(\begin{array}{cccc}
 c_3 & c_1-ic_2 & 0 & 0 \\
    c_1+ic_2 & -c_3 & 0 & 0 \\
    0 & 0 & -c_3 & -c_1+ic_2 \\
    0 & 0 & -c_1-ic_2 & c_3 \\
\end{array}\right)
\middle| c_1,c_2,c_3 \in \mathbb{R} \right\}.
\end{equation}

For the analysis of positive semidefiniteness, we recall the Schur complement criterion~\cite{Zhang} for a Hermitian block matrix. Consider the Hermitian block matrix
\begin{equation*}
 M = \left(\begin{array}{cc}
    A &B\\
    B^* & C
\end{array}\right),
\end{equation*}
where $A$ and $C$ are square matrices. Then $M \geq 0$ if and only if 
\begin{enumerate}
    \item $A \geq 0$,
    \item $\colsp(B) \subseteq \colsp(A)$,
    \item  $M/A \geq 0$.
\end{enumerate}
Here $\colsp(A)$ denotes the column space of $A$, $M /A = C - B^\dagger A^-B$, where $A^-$ is a generalized inverse of $A$, i.e., any matrix satisfying $A A^- A = A$. 

With these observations, we are ready to prove Theorem \ref{rank2}.
\begin{theorem}\label{rank2}
    A quantum channel $\Phi_0 \in \mathcal{C}han(2,2)$ of the form (\ref{matrixofrank2}) can be uniquely determined by the set of density matrices $\mathcal{S}_3$~(\ref{three_states}). 
\end{theorem}
\noindent{\bf Proof.}
Consider a quantum channel $\Phi_0$ with a matrix (\ref{matrixofrank2}) in the standard basis. Its Choi matrix has the form
 
\begin{equation*} 
\fl\quad \mathcal{C}_{\Phi_0}=\frac{1}{2}
 \left(\begin{array}{cccc}
  1 + \cos(u-v) & 0 & 0 & \cos{u} + \cos{v} \\
0 & 1 - \cos(u-v) & \cos{u} - \cos{v} & 0 \\
0 & \cos{u} - \cos{v} & 1 - \cos(u+v) & 0 \\
\cos{u} + \cos{v} & 0 & 0 & 1 + \cos(u+v)
\end{array}\right).
\end{equation*}
To establish when this matrix has rank 2, we  apply a permutation matrix $S$ to put $\mathcal{C}_{\Phi_0}$ in a block-diagonal form
\begin{equation*}
\fl    S^{-1}\mathcal{C}_{\Phi_0} S = \frac{1}{2}
\left(\begin{array}{cccc}
1 + \cos(u-v) & \cos{u}+\cos{v}& 0 & 0 \\
\cos{u}+\cos{v}& 1 + \cos(u+v) & 0 & 0 \\
0 & 0 & 1 - \cos(u-v) & \cos{u}-\cos{v} \\
0 & 0 & \cos{u}-\cos{v} & 1 - \cos(u+v)
\end{array}\right).
\end{equation*}
The determinant of each block $2\times2$ is zero, which means that each block has rank at most one. Consequently, the total rank of $\mathcal{C}_\Phi$ is at most two. The rank reduces to one precisely when one of the blocks is zero. The first block vanishes for $(u,v)=(\pi, 0)$, and the second for $(u,v)=(0,0)$. As our analysis focuses on channels with a Kraus rank of exactly two, we exclude these degenerate cases. 

Our goal is to prove that $(\mathcal{C}_{\Phi_0} + \mathcal{C}_W) \cap S_+(\mathbb{C}^4) = \{ \mathcal{C}_{\Phi_0}\}$, where $S_+(\mathbb{C}^4)$ denotes the cone of positive semidefinite complex matrices $4 \times 4$ and $\mathcal{C}_W$ has the form (\ref{W}). From Proposition~\ref{affinespacerho1rho2rho3} it follows that  $\mathcal{P}_\mathcal{C}(\mathcal{N}, \mathcal{S}_3) = \mathcal{C}_\mathcal{N}+\mathcal{C}_W$.
 For any $\mathcal{C}_A \in \mathcal{C}_W$, we show that if $(\mathcal{C}_{\Phi_0} + \mathcal{C}_A) \in S_+(\mathbb{C}^4)$, then $\mathcal{C}_A = 0$. 
 
Let $\{ e_1, e_2\}$ be an orthonormal basis of $\ker \mathcal{C}_{\Phi_0}$. We extend $\{ e_1, e_2\}$ to an orthonormal basis of $\mathbb{C}^4$. 
In this basis, the matrices become 
\begin{equation*}
    \tilde{\mathcal{C}}_{\Phi_0} =Q^{-1} \mathcal{C}_{\Phi_0}Q=
    \left(\begin{array}{cc}
 \mathbf{0}_{2\times 2} & \mathbf{0}_{2\times 2} \\
 \mathbf{0}_{2\times 2} & D
\end{array}\right)
\end{equation*}
and 
\begin{equation*}
    \tilde{\mathcal{C}}_A= Q^{-1} \mathcal{C}_{A}Q=
     \left(\begin{array}{cc}
 A_0 & A_1\\
    A_1^\dagger & A_2
\end{array}\right).
\end{equation*}
Here  $A_0, A_1, A_2,D \in \mathbb{C}^{2 \times 2}$ and $Q \in \mathbb{C}^{4 \times 4}$ is the transition matrix. Hence, 
\begin{equation*}
    \tilde{\mathcal{C}}_{\Phi_0}+ \tilde{\mathcal{C}}_{A} = 
     \left(\begin{array}{cc}
A_0 & A_1\\
    A_1^\dagger & A_2+D
\end{array}\right).
\end{equation*}

In our proof, we will use the following observation: if $A_0=0$ and $A_1 \neq 0$, then $\colsp(A_1) \not\subseteq \colsp(0) = \{0\}$; consequently, by the Schur complement criterion $\tilde{\mathcal{C}}_{\Phi_0}+ \tilde{\mathcal{C}}_{A}$ is not positive semidefinite. Taking this remark into account, we analyze three separate cases, depending on the structure of $\Phi_0$. 

   Case 1: Consider $u \neq \pi-v$ and $u \neq 2 \pi-v$. In this case,
\begin{equation*}
\fl\qquad  e_1 = \frac{1}{2|\cos \frac{u+v}{2}|\sqrt{1+\cos{u}\cos{v}}} \left(\begin{array}{cccc}
 - 1 - \cos (u+v), 0,0,\cos u + \cos v 
\end{array}\right)^T
\end{equation*}
and
\begin{equation*}
\fl\qquad e_2 = \frac{1}{2|\sin \frac{u+v}{2}|\sqrt{1-\cos{u}\cos{v}}} \left(\begin{array}{cccc}
0,- 1 + \cos (u+v), \cos u - \cos v,0
\end{array}\right)^T.
\end{equation*}
Since $(A_0)_{ij}=(e_i, \mathcal{C}_A e_j)$ for $i,j \in \{1,2\}$, by direct computation, we obtain
\begin{equation*}
A_0= \left(\begin{array}{cc}
 c_3&\Delta_1 c_1-i\Delta_2 c_2\\
        \Delta_1c_1+i\Delta_2 c_2&-c_3    
\end{array}\right),
\end{equation*} 
where 
\begin{equation*}
    \Delta_1=\frac{1- \cos^2(u+v)-\cos^2{u}+\cos^2{v}}{2|\sin(u+v)|\sqrt{1-\cos^2{u}\cos^2{v}}}
\end{equation*}
and 
\begin{equation*}
    \Delta_2=\frac{1- \cos^2(u+v)+\cos^2{u}-\cos^2{v}}{2|\sin(u+v)|\sqrt{1-\cos^2{u}\cos^2{v}}}.
\end{equation*}
If  $c_3 \neq 0$, then Sylvester's criterion for positive semi-definite matrices (see e.g.~\cite{meyer01}) impllies that $A_0$ is not positive semidefinite. Now suppose that $c_3=0$. In that case,
    \begin{equation*}
         \det A_0= -\Delta_1^2c_1^2-\Delta_2^2c_2^2 \leq 0.
    \end{equation*}
We distinguish four subcases:
\begin{enumerate}
        \item If $\Delta_1 = 0$ and $\Delta_2 = 0$, then from $\Delta_1+ \Delta_2=0$ it follows that $\cos(u+v)= \pm 1$. The solutions to these equations are $u=2 \pi - v$, $u=\pi - v$ and $(u, v)= (0, 0)$. 
        \item If  $\Delta_1 \neq 0$ and $\Delta_2 \neq 0$, then $A_0 \geq 0$ forces $c_1=0$ and $c_2=0$, yielding $\mathcal{C}_A=0$.
        \item If $\Delta_1 \neq 0$ and $\Delta_2 = 0$, then $A_0\geq 0$ requires $c_1=0$. Assuming $c_1=0$ and $c_2 \neq 0$, we have $A_0 = 0$. In this situation,
         \begin{equation*}
        \fl\qquad \mathcal{C}_{A}e_2=\frac{-ic_2}{2|\sin \frac{u+v}{2}|\sqrt{1-\cos{u}\cos{v}}} 
        \left(\begin{array}{cccc}
    - 1 + \cos (u+v), 0, 0,
    \cos{u}-\cos{v}
    \end{array}\right)^T.
        \end{equation*}
    By construction, the second column of $\tilde{\mathcal{C}}_A$ is precisely the representation of $\mathcal{C}_A e_2$ in the new basis. Thus, if $\mathcal{C}_A e_2 \neq 0$, then the second column of $\tilde{\mathcal{C}}_A$ is nonzero. Since $A_0 = 0$, the first two entries of this column are zero. Therefore, at least one of the last two entries must be nonzero, which implies $A_1^\dagger \neq 0$. Hence, $A_0=0$ and $A_1 \neq 0$, and by the Schur complement criterion $\mathcal{C}_A$ cannot be positive semidefinite. So $c_2=0$ and, hence, $\mathcal{C}_A=0$.
     \item If $\Delta_1 = 0$ and $\Delta_2 \neq 0$,  then $A_0\geq 0$ forces $c_2=0$. If $c_2=0$ and $c_1 \neq 0$, we obtain $A_0 = 0$ and
    \begin{equation*}
    \fl\qquad    \mathcal{C}_{A}e_2= \frac{c_1}{2|\sin \frac{u+v}{2}|\sqrt{1-\cos{u}\cos{v}}}
        \left(\begin{array}{cccc}
- 1 + \cos (u+v), 0, 0,
\cos{v}-\cos{u}
\end{array}\right)^T.
    \end{equation*}
again leading to $A_0=0$, $A_1 \neq 0$, and thus $\mathcal{C}_  A$ is not positive semidefinite. Consequently $c_1=0$ and $\mathcal{C}_A=0$.
\end{enumerate}

Case 2: Consider $u=2 \pi-v$, then
\begin{equation*}
    e_1=\frac{1}{\sqrt{1+\cos^2 v}} \left(\begin{array}{cccc}
1, 0, 0,-\cos{v}
\end{array}\right)^T, \quad e_2 = \left(\begin{array}{cccc}
  0, 0, 1, 0
\end{array}\right)^T.
\end{equation*}
Hence,
\begin{equation*}
A_0=  \left(\begin{array}{cc}
c_3 &(c_1 + i c_2) \frac{\cos v}{\sqrt{1+\cos^2v}} \\
(c_1 - i c_2) \frac{\cos v}{\sqrt{1+\cos^2v}}&-c_3 
\end{array}\right).
\end{equation*}
If $c_3 \neq 0$, $A_0$ is not positive semidefinite. If $c_3=0$, then
\begin{equation*}
    \det A_0 = -(c_1^2+c_2^2)\frac{\cos^2v}{1+\cos^2v} \leq 0.
\end{equation*}
We consider two possibilities:
\begin{enumerate}
    \item If $\cos v =0$, then $A_0=0$ and
    \begin{equation*}
        \mathcal{C}_{A}e_2 = (  0, 0, 0, -c_1-ic_2)^T.
    \end{equation*}
   $\mathcal{C}_Ae_4 \neq 0$ unless both $c_1$ and $c_2$ vanish; otherwise, by the Schur complement criterion, $\mathcal{C}_A$ is not positive semidefinite.
\item If $\cos v \neq 0$, then $A_0$ is positive semidefinite only when $c_1=0$ and $c_2=0$.

Case 3: Consider $u = \pi - v$, then 
\begin{equation*}
    e_1=\frac{1}{\sqrt{1+\cos^2 v}}\left(\begin{array}{cccc}
  0, 1, \cos v, 0
\end{array}\right)^T, \quad e_2 = \left(\begin{array}{cccc}
 0, 0, 0, 1
\end{array}\right)^T.
\end{equation*}
Consequently,
\begin{equation*}
A_0 = \left(\begin{array}{cc}
 -c_3 &-(c_1 - i c_2) \frac{\cos v}{\sqrt{1+\cos^2v}} \\ -(c_1 + i c_2) \frac{\cos v}{\sqrt{1+\cos^2v}}&c_3   
\end{array}\right).
\end{equation*}
Again, $c_3 \neq 0$ makes $A_0$ not positive semidefinite. For $c_3=0$ we have
\begin{equation*}
    \det A_0 = -(c_1^2+c_2^2)\frac{\cos^2v}{1+\cos^2v} \leq 0
\end{equation*}
Two subcases arise:
\begin{enumerate}
    \item $\cos v =0$ leads to $A_0=0$ and
    \begin{equation*}
        \mathcal{C}_{A}e_2 = \left(\begin{array}{cccc}
 0, 0, -c_1+ic_2, 0
\end{array}\right)^T
    \end{equation*}
   which is nonzero unless $c_1=c_2=0$, again violating the Schur complement criterion.
\item $\cos v \neq 0$ forces $c_1=c_2=0$ for $A_0$ to be positive semidefinite.
\end{enumerate}
In all cases, the positive semidefiniteness of $\mathcal{C}_{\Phi_0}+\mathcal{C}_A$ implies $c_1=c_2=c_3=0$. Therefore, we have shown that 
\begin{equation*}
    \mathcal{P}_{\mathcal{C}}(\Phi_0,\mathcal{S}_3) \cap S_+(\mathbb{C}^4) = \mathcal{C}_{\Phi_0},
\end{equation*}
 i.e., the channel $\Phi_0$ is uniquely determined by the set of states $\mathcal{S}_3$.
\end{enumerate}
This completes the proof.

\section{Generation of single-qubit quantum channels using coherent and incoherent controls}
\label{Problem_of_generating_Sec}

In the sections below, we investigate generation of non-unitary single-qubit quantum channels using coherent and incoherent controls. We consider the dynamics of an open two-level quantum system governed by a Gorini–Kossakowski–Sudarshan–Lindblad (GKSL) type equation with  a coherent  control $u(t)$ and an incoherent control $n(t)$:
\begin{equation}\label{master_equation}
\fl\qquad \frac{\mathrm{d}\rho^{u,n}_t}{\mathrm{d}t}  = \mathcal{L}^{u, n}_t\rho^{u,n}_t=- i [H_0 + V u(t), \rho^{u,n}_t] + \gamma \mathcal{D}_{n(t)}\rho^{u,n}_t,\quad \left.\rho^{u,n}_t\right|_{t=0}=\rho_0,
\end{equation}
where the dissipator is defined as follows:
\begin{eqnarray*} 
\fl\mathcal{D}_{n(t)}(\rho) =  (n(t)+1) \left( \sigma^-\rho\sigma^+ - \frac{1}{2}\{\sigma^+\sigma^- , \rho\}\right)+ n(t)\left(\sigma^+\rho\sigma^- -  \frac{1}{2}\{\sigma^-\sigma^+, \rho\}\right).
\end{eqnarray*}
Here 
$\sigma^+ = \left(\begin{array}{cc}
0 & 0 \\ 
1 & 0
\end{array}\right)
$
and
$
\sigma^- = \left(\begin{array}{cc}
0 & 1 \\ 
0 & 0
\end{array}\right)
$
are the raising and lowering matrices,
$[\cdot,\cdot]$ and $\{\cdot, \cdot\}$ stand for
the commutator and anti-commutator of two matrices, respectively.
The free Hamiltonian is defined as $H_0=\frac{1}{2}\omega(\mathbb{I}_2-\sigma_z)$, while the interaction Hamiltonian for the coherent control is $V=\mu\sigma_x$. Here, $\omega > 0$ is the transition frequency of the qubit, $\mu \ge 0$ is its coupling strength to the coherent control, and $\gamma \ge 0$ is the decoherence rate. We assume that the coherent  and incoherent controls are described by the functions from the class $L_\infty$. There are no amplitude restrictions on the function $u$, but the function $n$ is required to be nonnegative, i.e., $n \ge 0$.

Let $\Phi^{u,n}_t$ denote the evolution operator of the system: $\rho^{u,n}_t=\Phi^{u,n}_t \rho_0$.
For generation of a target quantum channel, we employ the dynamical equation for the system's evolution operator that follows from~(\ref{master_equation}): 
\begin{equation}
\label{master_equation_Phi}
\frac{\mathrm{d}\Phi^{u,n}_t}{\mathrm{d}t}  = \mathcal{L}^{u, n}_t\Phi^{u,n}_t,\quad \Phi^{u,n}_0=\mathbb{I}.
\end{equation}
The evolution operator generated by such dynamics is a quantum channel, i.e.,
$
\Phi^{u,n}_t\in\mathcal{C}han(2,2).
$

We formulate the problem of generation of a target quantum channel $\Phi_D\in\mathcal{C}han(2,2)$ as minimization of the following Mayer-type objective functional (dynamical objective functional) with a~fixed time~$T$:
\begin{equation}
    F_{\Phi_D}(u, n) = J_{\Phi_D}(\Phi^{u,n}_T) \to \inf_{u, n},
    \label{optimization_problem}
\end{equation}
where $J_{\Phi_D}\colon\mathcal{C}han(2,2)\to \mathbb{R}$ is a~functional on the set of quantum channels (kinematic objective functional)  which satisfies the following two conditions:
\begin{enumerate}
\item $J_{\Phi_D}(\Phi) \geq 0$; 
\item $J_{\Phi_D}(\Phi)=0$ if and only if $\Phi=\Phi_D$.
\end{enumerate}
The kinematic objective functional $J_{\Phi_D}$ is a~measure of the closeness of the quantum channel $\Phi$ to the quantum channel $\Phi_D$.  

Such a functional can be defined using some norm on the space of superoperators. We use the squared Frobenius norm (Hilbert--Schmidt norm)
$\|\Phi\|= \sqrt{\Tr (\Phi^\dagger \Phi)}$.
The natural kinematic objective functional is defined by the squared Hilbert--Schmidt distance between the target quantum channel $\Phi_D$ and a~quantum channel $\Phi$:
\begin{equation}\label{kinematic_func_Fr}
    J^{\rm Fr}_{\Phi_D}(\Phi)=\|\Phi-\Phi_{D}\|^2.
\end{equation}
The corresponding dynamical objective functional is:
\begin{equation}\label{dynamic_func_Fr}
    F^{\rm Fr}_{\Phi_D}(u,n)=J^{\rm Fr}_{\Phi_D}\bigl(\Phi_T^{u, n}\bigr)=\|\Phi_T^{u,n}-\Phi_{D}\|^2.
\end{equation}

Let $\mathcal{S}=\{\rho_j\}_{j=1}^d$ be a set of density matrices such that the quantum channel 
$\Phi_D$ is  uniquely determined by $\mathcal{S}$. 
Then we can introduce a kinematic objective functional defined as  the mean value of the Hilbert--Schmidt distance between the actions of $\Phi$ and $\Phi_D$ on the density matrices from the set $\mathcal{S}$:
\begin{eqnarray}\label{kinematic_func_st}
J_{\Phi_{D}}^{\mathcal{S}}(\Phi)=\frac{1}{d}\sum_{j=1}^d\|\Phi\rho_j - \Phi_{D}\rho_j\|^2.
\end{eqnarray}
The corresponding dynamical objective functional is:
\begin{equation}\label{states_norm}
    F_{\Phi_D}^{\mathcal{S}}(u,n)=J_{\Phi_D}^{\mathcal{S}}(\Phi_T^{u, n}) = \frac{1}{d}\sum_{j=1}^d\|\Phi_T^{u, n}\rho_j - \Phi_{D}\rho_j\|^2.
\end{equation}

\section{Gradient-Based Optimization Method}
\label{sec_inGRAPE}

For finding coherent and incoherent controls which generate a target quantum channel, we employ a first-order gradient based inGRAPE method \cite{PetruhanovPechenJPA2023}, developed for incoherent control and extending the GRAPE method \cite{Khaneja_Reiss_Kehlet_2005} to the general setting of coherent and incoherent control of open quantum systems. 

Computing the gradient of the objective functional in the Mayer problem of quantum channel optimization requires evaluating the variation of either the final system state $\rho_T$ or the matrix representation of the quantum channel $\Psi_T$. Therefore, we seek control functions in the form of simple piecewise-constant functions
\begin{eqnarray*}
    u(t)&=\sum_{j=1}^Nu_j\chi_{[t_{j-1};t_j)}, \quad u_j\in \mathbb{R},\\
    w(t)&=\sum_{j=1}^Nw_j\chi_{[t_{j-1};t_j)}, \quad w_j\in \mathbb{R},\\
    n(t)&=w^2(t)=\sum_{j=1}^Nw^2_j\chi_{[t_{j-1};t_j)}=\sum_{j=1}^Nn_j\chi_{[t_{j-1};t_j)}, \quad n_j\in \mathbb{R}^+,
\end{eqnarray*}
where $\chi_{[t_{j-1};t_j)}$ is the indicator function on the half-open interval ${[t_{j-1};t_j)}$, with $0<t_1<\dots<t_N=T$. The representation $n(t)=w^2(t)$ ensures the non-negativity condition $n(t) \ge 0$ automatically.

Here we introduce the dynamic functional $\widetilde{F}(v)=F(u(t),n(t))$, which depends on the vector $v=(u,w)^T=(u_1,\dots,u_N, w_1, \dots,w_N)^T \in \mathbb{R}^{2N}$ composed of the control values. At each iteration, the value of  the objective functionals is computed. If it is smaller than the value from the previous iteration, the new control is updated as:
$$v^{j+1}=v^j-h^j\mathrm{grad}_{v}\widetilde{F}(v).$$
Here $\mathrm{grad}_{v}\widetilde{F}(v)$ is a gradient of the objective functional.
The gradient factor $h^j$ is dynamic. In case of success, it is multiplied by some $a>1$; in case of failure, by some $b\in(0,1)$. The stopping criterion is either achieving an objective functional value $\widetilde{F}(v) < \epsilon$ or the absence of improvement over a maximum of $K_{\rm max}$ iterations. The exact expressions for the gradient of the objective functional are given in \ref{AppendixGradient}.

\section{Numerical experiments}
\label{Sec_numerical}

We consider generation of various rank-two qubit channels by coherent and incoherent controls  using the inGRAPE method with the objective functionals~(\ref{dynamic_func_Fr}) and~(\ref{states_norm}). For each channel generated using the state set functional~(\ref{states_norm}), the distance between the target and the generated channel is verified using the squared Frobenius norm-based functional~(\ref{dynamic_func_Fr}). On this basis, it is determined whether the channel has been successfully generated. We consider a quantum channel being successfully generated if its squared Frobenius norm distance from the target channel is less or equal to $10^{-2}$. In addition, generation in the absence of incoherent control has been also studied. In this mode, the initial incoherent control is set to zero everywhere and gradient with respect to incoherent control is not used for constructing controls for each iteration. Additionally, the quantum control landscape has been examined for the possible presence of traps, which are defined as local but not global minima of the landscape.

For the numerical experiments, the parameters of the master equation~(\ref{master_equation}) are chosen as $\omega = 1$, $\gamma = 0.01$, $\mu = 0.1$ to reflect a
typical relation that $\omega\gg \mu \gg \gamma$. Similar parameters were chosen in  \cite{PetruhanovPhotonics2023,PetruhanovPechenJPA2023}. The threshold stopping parameter of the optimization algorithm is $\epsilon=10^{-7}$, the maximal number $K_{\rm max}=20$, and the gradient multiplication parameters $a=1.1$, $b=0.5$ are used.

\subsection{Generation of the Replacement Channel}
Here we consider generation of a replacement quantum channel $\Phi_{\rho_D}$. 
If $\rho_D$ is a mixed quantum state, it can be represented as $\rho = p |\psi\rangle\langle\psi|+(1-p)|\psi^{\perp}\rangle\langle\psi^\perp|$, with $p\in (0,1)$. In this case, the Kraus representation of $\Phi_\rho$ is given by
$$
\Phi_\rho(\omega)=\sum_{k=1}^2K_{1k}\omega K_{1k}^\dagger+\sum_{k=1}^2K_{2k}\omega K_{2k}^\dagger,
$$
where the Kraus operators are $K_{1k} = \sqrt{p}|\psi\rangle\langle c_k|$ and $K_{2k} = \sqrt{1-p}  |\psi_\perp\rangle\langle c_k|$, $k=1,2$. If the target state $\rho_D$ is pure, i.e., $\rho_D = |\psi\rangle\langle\psi|$, the Kraus representation of $\Phi_{\rho_D}$ is $\Phi_{\rho_D}(\omega) = \sum_{k=1}^2 K_k \omega K_k^\dagger$, where the Kraus operators are $K_{k} = |\psi\rangle\langle c_k|$, $k=1,2$, and $\{|c_1\rangle, |c_2\rangle \}$ is an orthonormal basis in the Hilbert space of the system. The Kraus rank of the replacement quantum channel $\Phi_{\rho}$ is two if the target state $\rho_D$ is pure, and four if it is mixed.  If $\mathbf{r}_D$ is the Bloch vector of the target state $\rho_D$, then the matrix representation in the standard basis  of the replacement channel $\Phi_{\rho_D}$ takes the form:
\begin{equation*}
\mathbb{T}_{\Phi_{\rho_D}}=\left(\begin{array}{cc}
    1 &  \textbf{0}_{1\times3} \\
    \mathbf{r}_D & \textbf{0}_{3\times3}\end{array}\right).
\end{equation*}

\begin{figure}[t]
    \centering
        \includegraphics[width=\textwidth]{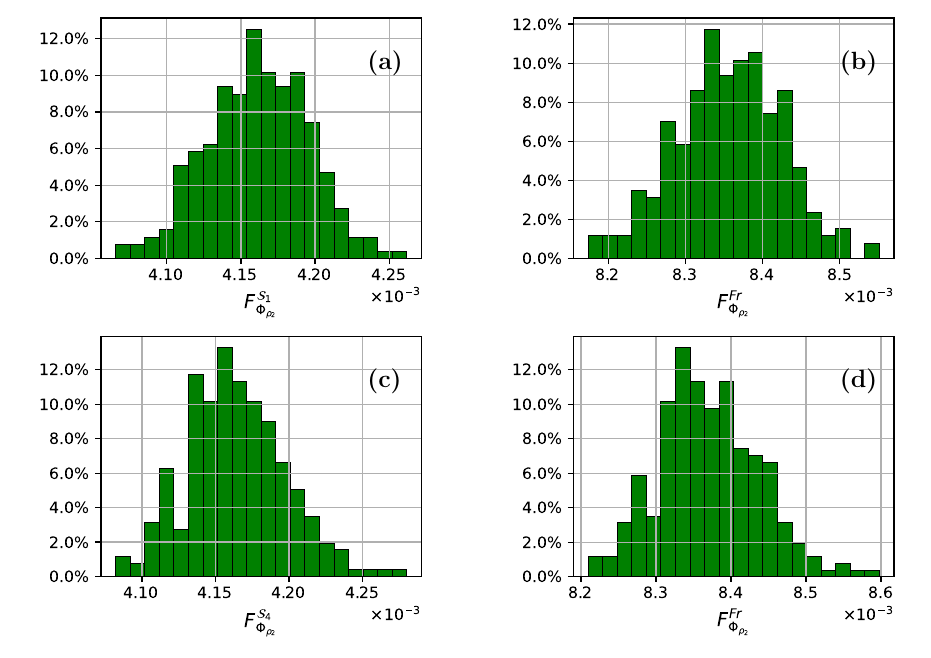}
    \caption{Histograms showing the distribution of final values of the objective functionals over 256 optimization runs for generation of the replacement channel $\Phi_{\bm{\rho}_2}$ using coherent and incoherent controls. The target pure state is $\bm{\rho}_2= |i\rangle\langle i|$ with Bloch vector $(0,1,0)^\top$.
    Histogram~(a) shows the optimization results for the functional $F^{\mathcal{S}_1}_{\Phi_{\bm{\rho}_2}}$ constructed with the set $\mathcal{S}_1$ consisting of a single chaotic state. Histogram~(c) corresponds to $F^{\mathcal{S}_4}_{\Phi_{\bm{\rho}_2}}$ constructed using the set $\mathcal{S}_4$ of four linearly independent density matrices.
    Right column: squared Frobenius distance from the obtained channel to the target channel $\Phi_{\bm{\rho}_2}$, evaluated via the functional $F^{\mathrm{Fr}}_{\Phi_{\bm{\rho}_2}}$ for  generation using the  functional $F^{\mathcal{S}_1}_{\Phi_{\bm{\rho}_2}}$~(b) and using the  functional $F^{\mathcal{S}_4}_{\Phi_{\bm{\rho}_2}}$~(d).}
\label{fig1:At1_hist}
\end{figure}

For the model under consideration, generation of replacement channels with a given accuracy is theoretically possible. Note that in the absence of control (i.e., when $u = n = 0$), the solution $\Phi_t^{0,0}$ of equation~(\ref{master_equation}) converges as $t\to\infty$ to the replacement channel $\Phi_{\bm{\rho}_4}$. State $\bm{\rho}_4=|0\rangle\langle 0|$ corresponds to the vector $(0,0,1)^T$ on the Bloch sphere which is the North Pole.
According to~\cite{Lokutsievskiy_Pechen_2021}, the reachable set from $\bm{\rho}_4$ under the master equation~(\ref{master_equation}) covers, with coherent control only, the entire Bloch ball with precision on the order of $\gamma/\omega$. More precisely, the closure of the reachable set spans the entire Bloch ball except for two lacunae of size on the order of $\gamma/\omega$; consequently, all replacement channels $\Phi_\rho$ can be approximately realized, with the only possible exceptions being those for which $\rho$ corresponds to a point lying within one of these lacunae. In our numerical experiments $\gamma/\omega = 0.01$, so the size of the lacunae is too small to affect the results of the generation of the replacement channels. 

\begin{figure}[t]
    \centering
        \includegraphics[width=\textwidth]{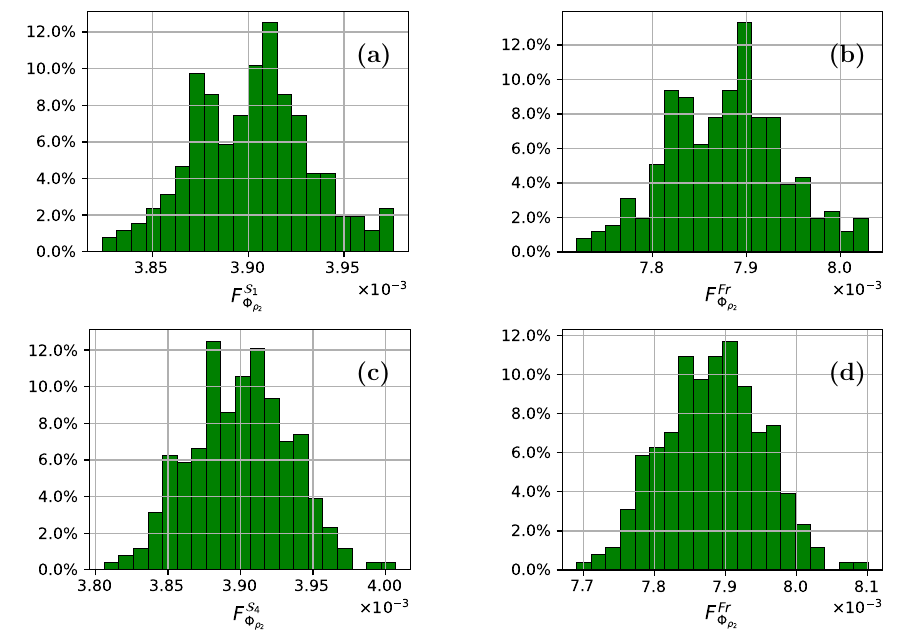}
    \caption{Histograms showing the distribution of final values of the objective functionals over 256 optimization runs for the generation of the replacement channel $\Phi_{\bm{\rho}_2}$ using only coherent control. Incoherent control is constant and equals zero. The target pure state is $\bm{\rho}_2= |i\rangle\langle i|$.    Histogram~(a) shows the optimization results obtained using the functional $F^{\mathcal{S}_1}_{\Phi_{\bm{\rho}_2}}$, which is generated by the set $\mathcal{S}_1$  consisting of a single chaotic state. Histogram~(c) corresponds to $F^{\mathcal{S}_4}_{\Phi_{\bm{\rho}_2}}$, generated by a set $\mathcal{S}_4$.
    Right column: squared Frobenius distance from the obtained channel to the target channel $\Phi_{\bm{\rho}_2}$, evaluated via the functional $F^{\mathrm{Fr}}_{\Phi_{\bm{\rho}_2}}$, for the generation using the  functional $F^{\mathcal{S}_1}_{\Phi_{\bm{\rho}_2}}$~(b) and using the  functional $F^{\mathcal{S}_4}_{\Phi_{\bm{\rho}_2}}$~(d).}
    \label{fig:At1_const_hist}
\end{figure}

To generate a replacement channel, a large evolution time is required; in the numerical experiments, the final time is chosen as $T=1500$ to allow any state sufficient time to purify under the free Hamiltonian. The control is discretized into $N=3000$ intervals. The inGRAPE method based on the objective functionals $F^{\mathcal{S}}_{\Phi_{\rho_0}}$, $\mathcal{S} \in \{\mathcal{S}_1, \mathcal{S}_4\}$, is used to find the controls. The set  $\mathcal{S}_1$ consists of the single chaotic state $\bm{\rho}_3=\frac{1}{2}\mathbb{I}$. The  set $\mathcal{S}_4$ consists of  four linearly independent pure states:\begin{eqnarray}\label{4_pure_state}
    \bm{\rho}_1=|+\rangle\langle+|=\frac{1}{2}\begin{pmatrixTWO}
        1 & 1 \\
        1 & 1
    \end{pmatrixTWO}, \quad \bm{\rho}_2=|i\rangle\langle i|=\frac{1}{2}\begin{pmatrixTWO}
        1 & -i \\
        i & 1
    \end{pmatrixTWO},\\
    \bm{\rho}_4=|0\rangle\langle0|=\begin{pmatrixTWO}
        1 & 0 \\
        0 & 0
    \end{pmatrixTWO}, \quad \bm{\rho}_5=|1\rangle\langle1|=\begin{pmatrixTWO}
        0 & 0 \\
        0 & 1
    \end{pmatrixTWO}\nonumber.
\end{eqnarray}

Under the action of the found controls, any state in this model first passes through the North Pole and then spirally approaches the target state with possible decoherence, as shown on Fig.~(\ref{fig:At1_wk}). This behavior makes the problem analogous to quantum state generation from the initial state being the North Pole. A direct consequence of this transit dynamics is that channels targeting states in the vicinity of the North Pole are achievable with greater accuracy at an order-of-magnitude coarser discretization.

\begin{figure}[t]
    \centering
        \includegraphics[width=\textwidth]{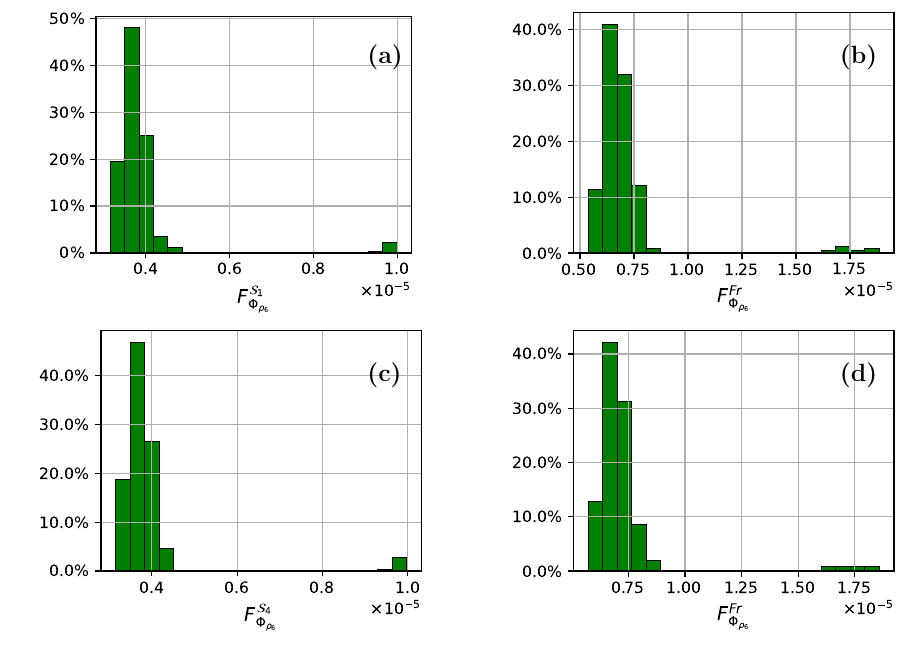}
    \caption{Histograms showing the distribution of final values of the objective functionals over 256 optimization runs for the generation of the replacement channel $\Phi_{\bm{\rho}_6}$ using coherent and incoherent control. The target state is mixed $\bm{\rho}_6= 3/4|i\rangle\langle i|+1/4|-i\rangle\langle -i|$ with Bloch vector $(0,0.5,0)^\top$.
    Histogram~(a) shows the optimization results obtained using the functional $F^{\mathcal{S}_1}_{\Phi_{\bm{\rho}_2}}$, which is generated by the set $\mathcal{S}_1$. Histogram~(c) corresponds to $F^{\mathcal{S}_4}_{\Phi_{\bm{\rho}_2}}$, generated by a set $\mathcal{S}_4$.
    Right column: squared Frobenius distance from the obtained channel to the target channel $\Phi_{\bm{\rho}_2}$, evaluated via the functional $F^{\mathrm{Fr}}_{\Phi_{\bm{\rho}_2}}$, for the generation using the  functional $F^{\mathcal{S}_1}_{\Phi_{\bm{\rho}_2}}$~(b) and using the  functional $F^{\mathcal{S}_4}_{\Phi_{\bm{\rho}_2}}$~(d).}
    \label{fig:At1_hist_ch}
\end{figure}

We investigate channels with various target states. For an illustration, the channel $\Phi_{\bm{\rho}_2}$ targeting $\bm{\rho}_2 = |i\rangle\langle i|$
(with the Bloch vector $=(0, 1, 0)^T$) was selected, as it is sufficiently distant from the North Pole to present a non-trivial generation problem. This quantum channel was generated via objective functional based on sets $\mathcal{S}_1$, $\mathcal{S}_4$ both in the presence and absence of incoherent control Fig.~(\ref{fig:At1_hist}) and Fig.~(\ref{fig:At1_const_hist}) correspondingly; the results are consistent across all functionals and insensitive to whether incoherent control is included or not. Furthermore, using the functional defined on a single chaotic state $\mathcal{S}_1$, a replacement channel $\Phi_{\bm{\rho}_6}$ targeting the mixed state $\bm{\rho}_6=\frac{3}{4}|i\rangle\langle i|+\frac{1}{4}|-i\rangle\langle -i|$ (with the Bloch vector $=(0, 0.5, 0)^T$) was successfully generated, as shown on Fig.~(\ref{fig:At1_hist_ch}). The results for this mixed-state target are better than for the pure-state target, owing to the closer proximity of $\bm{\rho}_6$ to the North Pole. Notably, despite the Kraus rank of this channel being four, generation has been achieved using a single state from the set. The incoherent control has virtually no effect on the generation of this process. Its values remain within the hypercube, while exhibiting a tendency to decrease during the motion from the North Pole toward the target state.

\begin{figure}[t]
    \centering
        \includegraphics[width=\textwidth]{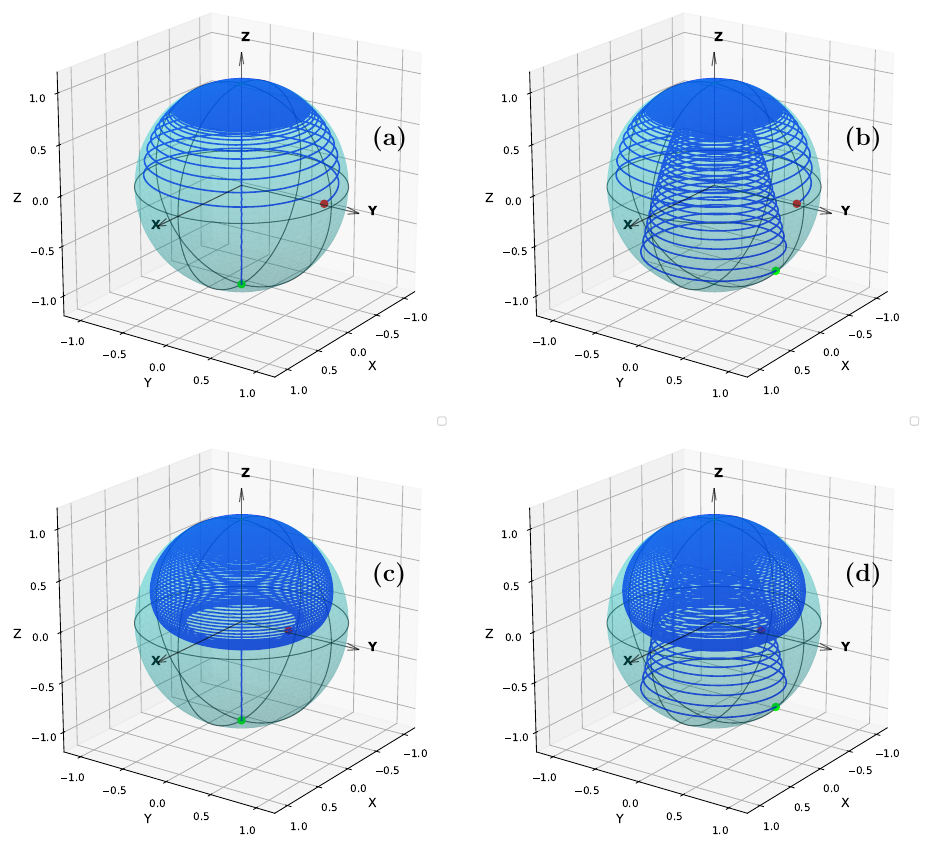}

    \caption{In the top row, the dynamics under one of the obtained controls are shown for generation using the objective functional $F^{\mathcal{S}_1}_{\Phi_{\rho_4}}$. In the bottom row, the dynamics under one of the obtained controls are shown for optimization using the objective functional $F^{\mathcal{S}_1}_{\Phi_{\rho_6}}$. Subplots (a, c) correspond to the initial state $r_0=(0,0,-1)^T$; subplots (b, d) correspond to $r_0=(0, 0.7, -0.7)^T$. The initial state is marked by green point and the final state by red point. All trajectories pass through the North Pole.}
    \label{fig:At1_wk}
\end{figure}

\FloatBarrier
\subsection{Phase Damping Channel}

The phase damping channel $\Phi_{\mathrm{damp}}(\omega)=K_0\omega K_0^\dagger+K_1\omega K_1^\dagger$ is defined by the following Kraus operators~\cite{IngermanZyczkowski}
\begin{equation*}
        K_0 = \begin{pmatrixTWO}
            1 & 0\\
            0 & \sqrt{1-p}
        \end{pmatrixTWO}, \ K_1 = \begin{pmatrixTWO}
            0 & 0\\
            0 & \sqrt{p}
        \end{pmatrixTWO}, \quad p\in(0,1].
\end{equation*}
The phase damping channel models the process of information loss due to the dephasing action of the environment, while the energy of the system remains unchanged. Under the action of this channel, for any nonzero $p$, the purity of the system strictly decreases, with lower purity corresponding to larger values of $p$. The matrix representation of this channel is $$\mathbb{T}_{\Phi_{\mathrm{damp}}}=\mathrm{diag}\left(1,\sqrt{1-p},\sqrt{1-p}, 1\right).$$
The phase damping channel admits the decomposition (\ref{rotation}) with 
\begin{equation*}
    U=e^{\frac{-i \pi \sigma_x}{4}}, V=e^{\frac{i \pi \sigma_x}{4}}, u = \arccos{\sqrt{1-p}}, v = 0.
\end{equation*}
By Proposition~\ref{unitinvar} and Theorem~\ref{rank2} it is uniquely determined by the set of states $\mathcal{S}^{'}_{3}= \{ V^\dagger \rho V | \rho \in \mathcal{S}_3\}$.
By direct computation, we obtain
$$\mathcal{S}^{'}_{3} = \{\bm{\rho}_1, \bm{\rho}_3, \bm{\rho}_4 \}.$$

\begin{figure}[t]
    \centering
        \includegraphics[width=\textwidth]{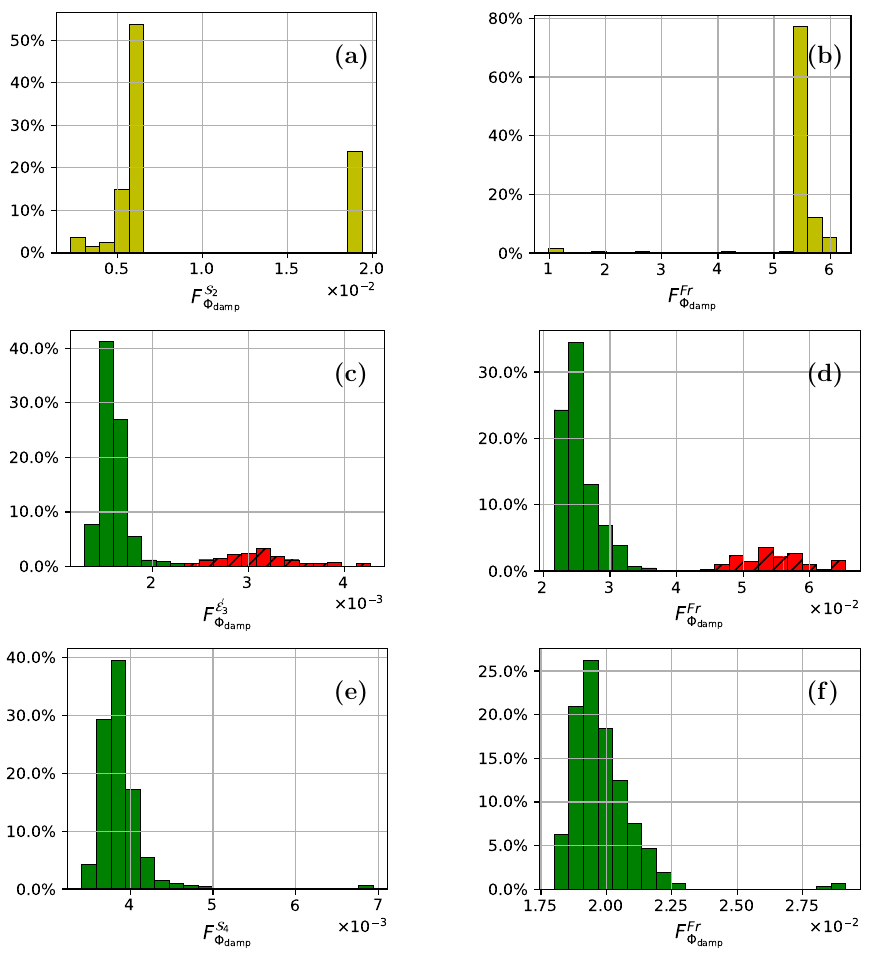}

    \caption{Histograms showing the distribution of final values of the objective functionals over 420 optimization runs for the generation of the phase damping channel $\Phi_{\mathrm{damp}}$ using coherent and incoherent control. Histogram~(a, c, e) shows the optimization results obtained using the functional $F^{\mathcal{S}}_{\Phi_{\mathrm{damp}}}$, which is generated by the sets $\mathcal{S}_2, \mathcal{S}_3^{'}, \mathcal{S}_4$  correspondingly.
    Histograms~(b, d, f) squared Frobenius distance from the obtained channel to the target channel $\Phi_{\mathrm{damp}}$, evaluated via the functional $F^{\mathrm{Fr}}_{\Phi_{\mathrm{damp}}}$, for the generation using the functional $F^{\mathcal{S}}_{\Phi_{\mathrm{damp}}}$~(\ref{dynamic_func_Fr}) with sets $\mathcal{S}_2, \mathcal{S}^{'}_3, \mathcal{S}_4$  correspondingly}
    \label{fig:ph_damping_hist}
\end{figure}

The channel is generated for the parameter $p=0.3$. For this value of the parameter $p$, the target channel is close to the identity channel which is unitary. Since the phase damping channel is determined by $\mathcal{E}^{'}_3$, the objective functional $F^{\mathcal{E}^{'}_3}_{\Phi_{\mathrm{damp}}}$ is sufficient for channel generation and it is confirmed numerically in our experiments. Furthermore, for comparison, this quantum channel was generated using the objective functional $F^{\mathcal{S}}_{\Phi_{\mathrm{damp}}}$ with an set of two states $\mathcal{S}_2=\{\bm{\rho}_3, \bm{\rho}_4\}$ and four states~$\mathcal{S}_4$~(\ref{4_pure_state}) (full basis of states), as well as the squared Frobenius norm-based functional $F^{\rm Fr}_{\Phi_{\mathrm{damp}}}$. The time and control discretization parameters were chosen as $T=15$, $K=50$. As can be seen from Fig.~(\ref{fig:ph_damping_hist}), the channel was generated using the functional~$F^{\mathcal{S}_3}_{\Phi_{\mathrm{damp}}}$ with the same accuracy as with the functional~$F^{\mathcal{S}_4}_{\Phi_{\mathrm{damp}}}$.

\begin{figure}[t]
    \centering
        \includegraphics[width=\textwidth]{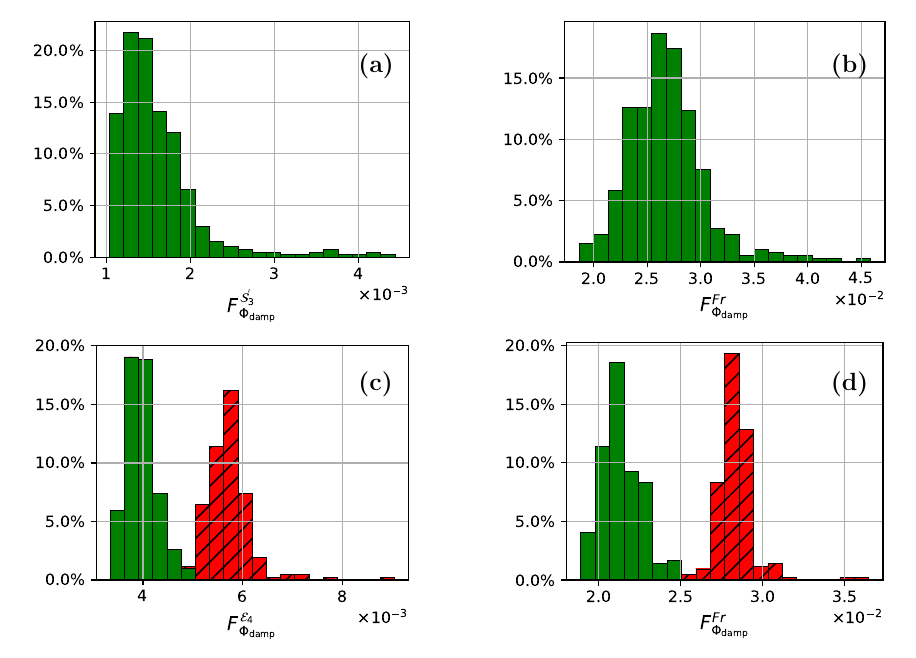}

    \caption{Histograms showing the distribution of final values of the objective functionals over 420 optimization runs for the generation of the phase damping channel $\Phi_{\mathrm{damp}}$ using only coherent control. Incoherent control is constant and equals zero. Histogram~(a, c) shows the optimization results obtained using the functional $F^{\mathcal{S}}_{\Phi_{\mathrm{damp}}}$, which is generated by the sets $\mathcal{S}_3^{'}, \mathcal{S}_4$  correspondingly.
    Histograms~(b, d) squared Frobenius distance from the obtained channel to the target channel $\Phi_{\mathrm{damp}}$, evaluated via the functional $F^{\mathrm{Fr}}_{\Phi_{\mathrm{damp}}}$, for the generation using the functional $F^{\mathcal{S}}_{\Phi_{\mathrm{damp}}}$ with sets $\mathcal{S}^{'}_3, \mathcal{S}_4$  correspondingly}
    \label{fig:ph_damping_hist_const}
\end{figure}

However, the objective functional~$F^{\mathcal{S}_2}_{\Phi_{\mathrm{damp}}}$, while yielding comparable final values~($10^{-3}$), fails to generate the channel. The distance to the target channel in the squared Frobenius norm is several orders larger than in the case of generation using the functionals~$F^{\mathcal{S}_3}_{\Phi_{\mathrm{damp}}}, F^{\mathcal{S}_4}_{\Phi_{\mathrm{damp}}}, F^{\rm Fr}_{\Phi_{\mathrm{damp}}}$. When the channel is successfully generated, the resulting controls split into two patterns, one of which corresponds to the peaks in~(\ref{fig:ph_damping_hist},~\ref{fig:ph_damping_hist_const}) with lower values of the objective functional (a possible trap). The optimization runs corresponding to this pattern are marked in green in Fig.~\ref{fig:ph_damping_ctrl}. The other pattern, marked in red, corresponds to the peak with higher final values of the objective functional and larger distance from the target channel to the generated one. It is notable that the channel is generated using the objective functional $F_{\Phi_{\mathrm{damp}}}^{\mathcal{S}^{'}_{3}}$. Moreover, in the absence of optimization over the incoherent control, the successful runs correspond only to the better, red peak. The same behavior is observed when optimizing $F_{\Phi_{\mathrm{damp}}}^{\mathcal{S}_{4}}$ using both coherent and incoherent controls. When generating with the objective functional $F_{\Phi_{\mathrm{damp}}}^{\mathcal{S}_{4}}$ in the absence of incoherent control optimization, the controls fall into one of the two patterns; the same occurs when generating with both coherent and incoherent controls using the functional $F_{\Phi_{\mathrm{damp}}}^{\mathcal{S}^{'}_{3}}$.

\begin{figure}[t]
    \centering
        \includegraphics[width=\textwidth]{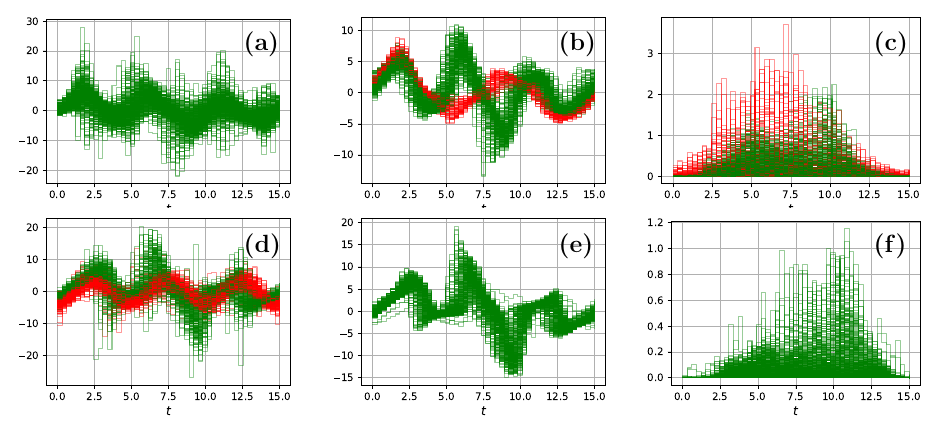}
    \caption{Control patterns obtained during the generation of the channel $\Phi_{\mathrm{damp}}$. Coherent controls (a, d) were generated using the functional $F_{\Phi_{\mathrm{damp}}}^{\mathcal{S}}$, where $\mathcal{S} = \mathcal{S}^{'}_3, \mathcal{S}_4$ respectively, in the absence of incoherent control. Coherent (b, e) and incoherent (c, f) controls were generated using the functional $F_{\Phi_{\mathrm{damp}}}^{\mathcal{S}}$, where $\mathcal{S} = \mathcal{S}^{'}_3, \mathcal{S}_4$ respectively, with optimization over both coherent and incoherent controls.}
    \label{fig:ph_damping_ctrl}
\end{figure}

\begin{figure}[t]
    \centering
        \includegraphics[width=\textwidth]{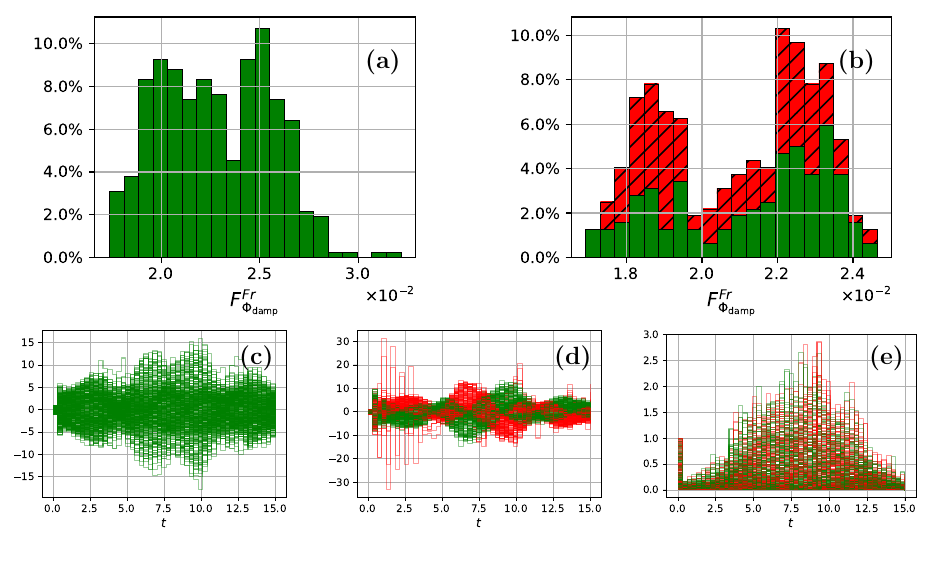}

    \caption{Histograms showing the distribution of final values of the objective functionals over 420 optimization runs for the generation of the phase damping channel $\Phi_{\mathrm{damp}}$ using only coherent (a) and both coherent and incoherent (b) control using the functional $F^{\rm Fr}_{\Phi_{\mathrm{damp}}}$. Subplots (c, d) show the coherent controls obtained by optimization using only coherent control and using both coherent and incoherent controls, respectively, while subplot (e) shows the incoherent controls.}
    \label{fig:ph_damping_hist_fr}
\end{figure}

\section{Conclusions and Discussion}
\label{Conclusions}
In this work, we consider the problem of controlled generation of a target quantum evolution (defined as a quantum channel) in an open quantum system. It is important that we consider in general non-unitary target quantum channels, and concentrate in more details on rank two single-qubit quantum channels as the simplest beyond unitary. For them, we study how many quantum states are sufficient to determine whether a general single-qubit quantum channel, representing some actual evolution of an open qubit, coincides with a rank two single-qubit quantum channel or not. For a unique determination of a generic single qubit quantum channel one has to compare its action on four density matrices. However, we prove that for determining if a generic single-qubit quantum channel coincides with a rank two quantum channel, comparing its action on only three special density matrices is always sufficient. 

This result is important for defining minimal in the sense of the number of density matrices objective functional for optimal generation of rank two single qubit quantum channels. Finding the minimal number of states necessary to uniquely determine a quantum channel is even more important in experimental settings to minimize the cost of overall quantum channel tomography, which is cost of tomography of one quantum state multiplied by the number of quantum states.

Thereby we use this result for constructing  objective functional to numerically investigate generation of single-qubit rank two quantum channels in physical models of qubit driven by coherent and incoherent controls. We consider as target quantum channels various replacement and phase damping channels. For optimization, we use inGRAPE method based on the explicit form of gradient for the objective functional constructed using the three special density matrices. The obtained minimal infidelities are different for different channels. For phase damping channel and replacement channels projecting on some target states, minimal found infidelity for minimizing the objective functional based on three states is about $10^{-3}$, while for replacement channels projecting on some other target states minimal infidelity is as small as $10^{-5}$. This difference might be explained by the possible lack of controllability of the considered open qubit in the set of all quantum channels. Such an interesting and important problem of controllability of open quantum systems in the set of quantum channels (including non-unitary) still requires its investigation. Establishing controllability of open quantum systems in the set of quantum channel, and extending the obtained results beyond two-level systems and beyond rank two quantum channels, are interesting directions for a future research.

\section*{Data availability statement}
All data that support the findings of this study are included within the article (and any supplementary
files).

\ack This work was funded by the Ministry of Science and Higher Education of the Russian Federation (grant number 075-15-2020-788).

\appendix
\section{Parametrization}
\label{Parametrization}
We consider the parametrization of the quantum state  in terms of the Bloch vector given by~(\ref{Bloch_parametrization}). In this parametrization, the master equation~(\ref{master_equation}) takes the form
\begin{equation}\label{master_equation_r}
\fl \qquad    \frac{\mathrm{d} \mathbf{r}}{\mathrm{d} t} = A(u(t), n(t))\mathbf{r}  + \mathbf{b}  = (B + B_uu(t)+B_nn(t))\mathbf{r}  + \mathbf{b}, \quad \mathbf{r}(0)=\mathbf{r}_0.
\end{equation}
Here
\begin{eqnarray*}
   B=\left(\begin{array}{ccc}
        -\frac{\gamma}{2} & \omega & 0 \\
        -\omega & -\frac{\gamma}{2} & 0 \\
        0 & 0 & -\gamma
    \end{array}\right), \quad B_u=\left(\begin{array}{ccc}
        0 & 0 & 0 \\
        0 & 0 & -2\mu \\
        0 & 2\mu & 0
    \end{array}\right),
    \\
    B_n=\left(\begin{array}{ccc}
        -\gamma & 0 & 0 \\
        0 & -\gamma & 0 \\
        0 & 0 & -2\gamma
    \end{array}\right), \quad \mathbf{b}=\left(\begin{array}{c}
        0\\
        0\\
        \gamma
    \end{array}\right).
\end{eqnarray*}

As mentioned in  Section~\ref{sec_inGRAPE}, the control is specified as a piecewise-constant function associated with vectors $v=(u, w)\in\mathbb{R}^{2N}$. So we denote the value of the operator $A(u(t), n(t))$ on the interval $[t_{k-1}, t_k)$ as $A_k$. Since the control is piecewise-constant, we can obtain the solution to equation~(\ref{master_equation_r}) as a product of matrix exponentials:
\begin{eqnarray}\label{r_solution}
    \mathbf{g}_k=(e^{A_k \Delta t}-\mathbb{I})A_k^{-1}\mathbf{b},\\
	\mathbf{r}(t_k)\equiv\mathbf{r}_k=e^{A_k \Delta t}\mathbf{r}_{k-1}+\mathbf{g}_k,  \quad \mathbf{r}(0)=\mathbf{r}_0	\quad k=1,\ldots N\nonumber.
\end{eqnarray}
We denote $\mathbf{r}_T(v)=\mathbf{r}_N$.

Let $\Phi_D$ be the target channel and let $\mathcal{S}=\{\rho_j\}_{j=1}^d$ be a set of density matrices that uniquely determines this  channel.  For each $j \in \{1,\dots,d\}$, let $\mathbf{r}_j$ denote the Bloch vector corresponding to the state $\rho_j \in \mathcal{S}$, let $\mathbf{r}_j^T(v)$ be the solution to equation~(\ref{r_solution}) at time $T$ with initial condition $\mathbf{r}_j$, and let $\mathbf{r}_j^D$ denote the Bloch vector of the state $\rho^D_j = \Phi_D \rho_j$. In this parametrization, the objective functional $F^{\mathcal{S}}_{\Phi_D}$~(\ref{states_norm}) takes the form
\begin{equation}\label{dynamic_func_st_r}
    F^{\mathcal{S}}_{\Phi_D}(v)= \frac{1}{2d}\sum_{j=1}^d\|\mathbf{r}^j_T(v) - \mathbf{r}^j_{D}\|^2.
\end{equation}

To obtain the  functional $F^{\rm Fr}_{\Phi_D}$~(\ref{dynamic_func_Fr}), it is necessary to parametrize equation~(\ref{master_equation_Phi}). We introduce $\mathbf{q}=(1,\mathbf{r})^T$ to bring the master equation~(\ref{master_equation_r}) into homogeneous form. It then reads

\begin{equation}\label{master_equation_q}
\fl\qquad\frac{\mathrm{d} \mathbf{q}}{\mathrm{d}t} = C(u(t), n(t))\mathbf{q} = \left(\begin{array}{cc}
0 & \mathbf{0}_{1\times3} \\
\mathbf{b} & A(u(t), n(t))\end{array}\right)\mathbf{q},\quad\mathbf{q}(0)=\mathbf{q}_0=(1,\mathbf{r}_0)^T.
\end{equation}
For the quantum channel $\Phi^{u,n}_t$ denote $\Psi(t)=\mathbb{T}_{\Phi^{u,n}_t}$.  Then
\begin{equation*}
\mathbf{q}(t) = \Psi(t)  \mathbf{q}(0),
\end{equation*}

\begin{equation}\label{master_equation_psi}
    \frac{\mathrm{d} \Psi}{\mathrm{d}t} = C(u(t), n(t))\Psi, \quad \Psi(0) = \mathbb{I}.
\end{equation}
  Since we consider   piecewise-constant   controls, the solution can be written as:
\begin{equation}\label{Psi_solution}
     \Psi_k = e^{C_k\Delta t} \Psi_{k-1}, \quad \Psi_0=
    \mathbb{I}, \quad k=1,\ldots N.
\end{equation}
Here $
    C_k = \left(\begin{array}{cc}
0 & \mathbf{0}_{1\times3} \\
\mathbf{b} & A_k\end{array}\right)
$. We denote $\Psi_T(v)=\Psi_N$.

We can obtain the explicit matrix form of the desired channel $\Phi_D$ to derive the explicit form of the functional $F^{\rm Fr}_{\Phi_D}$~(\ref{dynamic_func_Fr}) using its Kraus representation 
\begin{equation*}
    \Phi_D\rho = \sum_{m=1}^l
    K_m\rho K_m^\dagger=\sum_{m=1}^{l} K_m\left(\frac{1}{2}\sum_{j=0}^3\sigma_j q_j\right)K_m^\dagger,
\end{equation*}
where $\mathbf{q}=(q_0,q_1,q_2,q_3)^T=(1,\mathbf{r})^T$.
Multiplying successively by all Pauli matrices and taking the trace, we obtain
\begin{equation*}
\fl\qquad    \mathbf{q}_{(D)i}=\Tr\left(\sum_{m=1}^lK_m\left(\frac{1}{2}\sum_{j=0}^3\sigma_j \mathbf{q}_j\right)K_m^\dagger\sigma_i\right)=\sum_{j=0}^3\frac{1}{2}\Tr\left(\sum_{m=1}^{l}K_m\sigma_jK_m^\dagger\sigma_i\right)\mathbf{q}_j.
\end{equation*}
Then the matrix elements of the channel representation are given by:
\begin{equation}
\label{channel_parametrization}
    \left[\Psi_D\right]_{ij}=\frac{1}{2}\Tr\left(\sum_{m=1}^lK_m\sigma_jK_m^\dagger\sigma_i\right).
\end{equation}
The objective functional $\widetilde{F}^{\rm Fr}_{\Phi_D}$ (\ref{dynamic_func_Fr}) can be expressed  in the form 
\begin{equation}\label{dynamic_func_Fr_Psi}
    \widetilde{F}^{\rm Fr}_{\Phi_D}(v)=\|\Psi_D-\Psi_T(v)\|^2.
\end{equation}

\section{Analytical formulas for Gradient}
\label{AppendixGradient}

To employ gradient-based optimization methods, it is necessary to compute the gradient of the functional with respect to the vector formed by the values of the coherent and incoherent controls. We compute the gradients of the functionals $\widetilde{F}_{\Phi_D}^\mathcal{S}$ (\ref{dynamic_func_st_r}) and $\widetilde{F}_{\Phi_D}^{\rm Fr}$ (\ref{dynamic_func_Fr_Psi}) with respect to the vector $v=(u, w)^T$: 
\begin{equation}\label{gradient_F_St}
\fl\qquad
    \frac{\partial}{\partial v_k}\widetilde{F}^{\mathcal{S}}_{\Phi_D}(v) = \frac{\partial}{\partial v_k}  \frac{1}{2d}\sum_{j=1}^d\|\mathbf{r}^j_T(v) - \mathbf{r}^j_{D}\|^2  = \frac{1}{d}\sum_{j=1}^d\langle \mathbf{r}^j_T(v) - \mathbf{r}^j_{D}, \frac{\partial}{\partial v_k}\mathbf{r}^j_T(v)\rangle,
\end{equation}

\begin{equation}\label{gradient_F_Gr}
\fl\qquad
    \frac{\partial}{\partial v_k}\widetilde{F}^{\rm Fr}_{\Phi_D}(v) = \frac{\partial}{\partial v_k} \|\Psi_T(v) - \Psi_{D}\|^2  = 2\Tr\left((\Psi_T(v) - \Psi_{D})^\dagger \frac{\partial}{\partial v_k}\Psi_T(v)\right).
\end{equation}
These formulas were introduced in Theorem~3 of~\cite{PetruhanovPechenJPA2023}. To find the gradient $\frac{\partial}{\partial v_k}\mathbf{r}^i_T(v)$ in~(\ref{gradient_F_St}), it is necessary to differentiate the solution of equation~(\ref{r_solution}). For the derivative of the exponential of a non-commuting object, we  use the well-known integral formula
\begin{eqnarray*}
\fl\qquad
   \frac{\partial}{\partial u_k}\mathbf{r}_T(v)=\frac{\partial}{\partial u_k}\mathbf{r}_T(u, w) = e^{A_N \Delta t_N} \dots e^{A_{k + 1} \Delta t_{k + 1}}\left(\left[\frac{\partial}{\partial u_k}e^{A_k \Delta t_k}\right]\mathbf{r}_{k-1} +  \frac{\partial}{\partial u_k} {\mathbf{g}_k}\right),
\end{eqnarray*}
\begin{eqnarray*}
\fl\qquad
    \frac{\partial}{\partial w_k}\mathbf{r}_T(v)=  \frac{\partial}{\partial w_k}\mathbf{r}_T(u,w) = e^{A_N \Delta t_N} \dots e^{A_{k + 1} \Delta t_{k + 1}}\left(\left[\frac{\partial}{\partial w_k}e^{A_k \Delta t_k}\right]\mathbf{r}_{k-1} +  \frac{\partial}{\partial w_k} {\mathbf{g}_k}\right).
\end{eqnarray*}
Here
\begin{eqnarray*}
    \frac{\partial}{\partial u_k}e^{A_k \Delta t_k}=\Delta t_k\int \limits _0^1 e^{\alpha A_k \Delta t_k } B_u e^{(1-\alpha)A_k \Delta t_k} \rmd\alpha,\\
    \frac{\partial}{\partial w_k}e^{A_k \Delta t_k}=2\Delta t_kw_k\int \limits _0^1 e^{\alpha A_k \Delta t_k } B_n e^{(1-\alpha)A_k \Delta t_k} \rmd\alpha,\\
    \frac{\partial}{\partial u_k}\mathbf{g}_k = \left(\frac{\partial}{\partial u_k}e^{A_k \Delta t_k} - (e^{A_k \Delta t_k} - I)A_k^{-1} B^u\right) A_k^{-1} \mathbf{b},\\
    \frac{\partial}{\partial w_k}\mathbf{g}_k = \left(\frac{\partial}{\partial w_k}e^{A_k \Delta t_k} - 2w_k(e^{A_k \Delta t_k} - I)A_k^{-1} B^n\right) A_k^{-1} \mathbf{b}.
\end{eqnarray*}
To obtain the gradient in~(\ref{gradient_F_Gr}), it is necessary to differentiate~(\ref{Psi_solution}), which requires the use of the integral differentiation formula:
\begin{eqnarray*}
\fl\qquad
    \frac{\partial}{\partial u_k}\Psi_T(u, w) = e^{C_N\Delta t}\dots e^{C_{k+1}\Delta t} \left(\int_0^1 \Delta t\, e^{C_k\Delta t (1-\alpha)}\left(\begin{array}{cc}
    0 &  \mathbf{0}_{1\times3} \\
    \mathbf{0}_{3\times1} & B_u\end{array}\right)e^{C_k\alpha\Delta t } \rmd\alpha\right) \Psi_{k-1},\\
\fl\qquad
    \frac{\partial}{\partial w_k}\Psi_T(u, w) = e^{C_N\Delta t}\dots e^{C_{k+1}\Delta t} \left(\int_0^1 2\Delta t w_k\, e^{C_k\Delta t (1-\alpha)}\left(\begin{array}{cc}
    0 &  \mathbf{0}_{1\times3} \\
    \mathbf{0}_{3\times1} & B_n\end{array}\right)e^{C_k\alpha\Delta t } \rmd\alpha\right) \Psi_{k-1}.
\end{eqnarray*}

\end{document}